\newif\ifConference
\newif\ifJournal
\newif\ifArXiv

\ArXivtrue

\ifConference
	
	\documentclass[runningheads]{llncs}

\usepackage[utf8]{inputenc}
\usepackage{amsmath}
\usepackage{amssymb}
\usepackage{amsthm}
\usepackage{enumerate}
\usepackage{mathtools}
\usepackage{graphicx}
\usepackage{tikz}
\usetikzlibrary {arrows.meta,backgrounds,calc,decorations.markings,decorations.pathreplacing,fit,graphs}
\usepackage{nicefrac}
\usepackage{tabularx}
\usepackage{xspace}
\usepackage[T1]{fontenc}
\usepackage{hyperref}
\usepackage{paralist}
\usepackage{booktabs}
\hypersetup{
	colorlinks,
	linkcolor={red!50!black},
	citecolor={blue!90!black!80},
	urlcolor={blue!80!black}
}
\usepackage{cleveref}

\usepackage[disable]{todonotes}
\usepackage{lineno}

{\upshape\itshape}{\upshape\rmfamily}

{\upshape\itshape}{\upshape\rmfamily}
\ifArXiv
\else
\newtheorem{observation}{Observation}{\upshape\itshape}{\upshape\rmfamily}
\fi
{\upshape\itshape}{\upshape\rmfamily}
\newtheorem{rrule}{Reduction Rule}{\upshape\itshape}{\upshape\rmfamily}
\ifArXiv
\crefname{section}{Section}{Sections}
\crefname{lemma}{Lemma}{Lemmata}
\crefname{observation}{Observation}{Observations}
\crefname{theorem}{Theorem}{Theorems}
\crefname{rrule}{Rule}{Rules}
\crefname{figure}{Figure}{Figures}
\else
\crefname{section}{Sec.}{Secs.}
\crefname{lemma}{Lem.}{Lems.}
\crefname{observation}{Obs.}{Obs.}
\crefname{theorem}{Thm.}{Thms.}
\crefname{rrule}{Rule}{Rules}
\crefname{figure}{Fig.}{Figs.}
\fi

\newcommand{\kommentar}[1]{}
\newcommand{\Oh}{\ensuremath{\mathcal{O}}}

\DeclareMathOperator{\off}{off}

\newcommand{\Dbar}{{\ensuremath{\overline{D}}}\xspace}
\newcommand{\pDbar}{{\ensuremath{\overline{D'}}}\xspace}
\newcommand{\kbar}{{\ensuremath{\overline{k}}}\xspace}
\newcommand{\pkbar}{{\ensuremath{\overline{k'}}}\xspace}
\newcommand{\w}{{\ensuremath{\omega}}\xspace}

\newcommand{\predatorsF}[2]{{\ensuremath{N^{>}_{#2}(#1)}}\xspace}
\newcommand{\preyF}[2]{{\ensuremath{N^{<}_{#2}(#1)}}\xspace}

\newcommand{\prey}[1]{\preyF{#1}{}}

\newcommand{\yes}{{\normalfont\texttt{yes}}\xspace}

\newcommand{\Wh}[1]{{\normalfont W[#1]}\xspace}
\newcommand{\NP}{{\normalfont{NP}}\xspace}

\newcommand{\FPT}{{\normalfont{FPT}}\xspace}

\newcommand{\XP}{{\normalfont{XP}}\xspace}
\newcommand{\NPcoNPpoly}{{\normalfont{NP~$\not\subseteq$~coNP/poly}}\xspace}

\newcommand{\Instance}{{\ensuremath{\mathcal{I}}}\xspace}

\newcommand{\Tree}{{\ensuremath{\mathcal{T}}}\xspace}
\newcommand{\Food}{{\ensuremath{\mathcal{F}}}\xspace}

\newcommand{\PD}{\PDsub\Tree}
\newcommand{\PDsub}[1]{{\ensuremath{PD_{#1}}}\xspace}

\newcommand{\A}{\mathcal{A}}

\newcommand{\NN}{\mathbb{N}}

\newcommand{\problemdef}[3]{
  \begin{trivlist}
    \item \tikz{\node [draw,inner sep=4.5pt] {\begin{minipage}{\textwidth-9.4pt}
          \normalsize\textsc{#1}
          
          \smallskip

          \begin{tabularx}{\textwidth-9.4pt}{ll@{\hspace{3pt}}>{\raggedright}X}
            \normalsize\textbf{Input} & :	& \normalsize#2 \cr
            \normalsize\textbf{Question} & :				& \normalsize#3
         \end{tabularx}
    \end{minipage}};}
  \end{trivlist}}

\newcommand{\PROB}[1]{{{\textsc{#1}}}\xspace}

\newcommand{\MPDlong}{\PROB{Maximize Phylogenetic Diversity}}

\newcommand{\CP}{\PROB{Clique}}

	\usepackage{hyperref}
	\author{
			Niels Holtgrefe${}^{1,}$
			\thanks{
				Funded by the Dutch Research Council (NWO), grant OCENW.M.21.306.
			},
			Jannik Schestag${}^{1,}$
			\thanks{
				Funded by NWO, grant OCENW.GROOT.2019.015.
			}, and
			Norbert Zeh${}^{2,}$
			\thanks{
				Research supported by NSERC Discovery Grant number RGPIN-2025-06235.
			}
	}
	\authorrunning{Holtgrefe, Schestag, and Zeh}
	\institute{
			${}^{1}$ TU Delft, Delft, The Netherlands.
			${}^{2}$ Dalhousie University, Halifax, Canada.
	}
\else
	
	\documentclass[a4paper,USenglish,cleveref,autoref,thm-restate,numberwithinsect,nameinlink]{lipics-v2021}

	\author{Niels Holtgrefe}
	{Delft University of Technology, Delft, The Netherlands}
	{n.a.l.holtgrefe@tudelft.nl}
	{https://orcid.org/0009-0001-6162-9668}
	{Funded by the Dutch Research Council (NWO), grant OCENW.M.21.306.}
	\author{Jannik Schestag}
	{Delft University of Technology, Delft, The Netherlands}
	{j.t.schestag@tudelft.nl}
	{https://orcid.org/0000-0001-7767-2970}
	{Funded by NWO, grant OCENW.GROOT.2019.015,
		``Optimization for and with Machine Learning (OPTIMAL).''}
	\author{Norbert Zeh}
	{Dalhousie University, Halifax, Canada}
	{nzeh@cs.dal.ca}
	{https://orcid.org/0000-0002-0562-1629}
	{Research supported by NSERC Discovery Grant number RGPIN-2025-06235.}
	
	
	\authorrunning{Holtgrefe, Schestag, and Zeh}
	
	
	\hideLIPIcs
\fi

\usetikzlibrary{shadows}

\newcommand{\PDDlong}{\PROB{Optimizing PD with Dependencies}}

\newcommand{\PDD}{\PROB{\mbox{$\varepsilon$-PDD}}}
\newcommand{\sPDD}{\PROB{\mbox{$\varepsilon$-PDD$_{\text{s}}$}}}

\newcommand{\fPDD}{\PROB{\mbox{$1$-PDD}}}
\newcommand{\fsPDD}{\PROB{\mbox{$1$-PDD$_{\text{s}}$}}}

\newcommand{\hPDD}{\PROB{\mbox{$\nicefrac{1}{2}$-PDD}}}
\newcommand{\hsPDD}{\PROB{\mbox{$\nicefrac{1}{2}$-PDD$_{\text{s}}$}}}

\newcommand{\aPDDlong}{\PROB{Optimizing PD with $\alpha$-Dependencies}}
\newcommand{\aPDD}{\PROB{\mbox{$\alpha$-PDD}}}
\newcommand{\asPDD}{\PROB{\mbox{$\alpha$-PDD$_{\text{s}}$}}}

\newcommand{\wPDD}{\PROB{Weighted PDD}}

\newcommand{\rwPDD}{\PROB{\mbox{rw-PDD}}}

\newcommand{\viable}{{$\varepsilon$-viable}\xspace}
\newcommand{\fviable}{{$1$-viable}\xspace}

\newcommand{\aviable}{{$\alpha$-viable}\xspace}

\DeclareMathOperator{\vc}{vc}

\newcommand{\todos}[2][]{\todo[#1,color=red!25!green!50]{ #2}}
\newcommand{\todosi}[2][]{\todo[inline,color=red!25!green!50]{ #2}}

\usepackage{hyphenat}
\allowdisplaybreaks

\title{Limits of Kernelization and Parametrization for Phylogenetic Diversity with Dependencies
	\footnote{This paper was accepted for an oral presentation at the 17th Latin American Theoretical Informatics Symposium (LATIN 2026) in Florianópolis, Brazil, April 2026.}}

\titlerunning{Kernelized and Parameterized Complexity of \aPDD}

\nolinenumbers

\begin{document}

\maketitle
\todos{I suggest kicking the unnecessary line of electronic mail for the proceedings.}

\begin{abstract}
	In the \MPDlong problem, we are given a phylogenetic tree that represents the genetic proximity of species, and we are asked to select a subset of species of maximum phylogenetic diversity to be preserved through conservation efforts, subject to budgetary constraints that allow only $k$~species to be saved.
	This neglects that it is futile to preserve a predatory species if we do not also preserve at least a subset of the prey it feeds on.
	Thus, in the \PDDlong (\PDD) problem, we are additionally given a food web that represents the predator-prey relationships between species.
	The goal is to save a set of $k$~species of maximum phylogenetic diversity such that for every saved species, at least one of its prey is also saved.
	This problem is \NP-hard even when the phylogenetic tree is a star.
	
	The \aPDD problem alters \PDD by requiring that at least some fraction $\alpha$ of the prey of every saved species are also saved.
	In this paper, we study the parameterized complexity of \aPDD.
	We prove that the problem is \Wh{1}-hard and in~\XP when parameterized by the solution size~$k$, the diversity threshold~$D$, or their complements.
	When parameterized by the vertex cover number of the food web, \aPDD is fixed-parameter tractable (\FPT).
	A~key measure of the computational difficulty of a problem that is \FPT is the size of the smallest kernel that can be obtained.\todo{JS: Version suggested by Niels which leaves out the fact that it is parameterized. I am fine with that. The previous version is commented out.}
	We prove that, when parameterized by the distance to clique, \fPDD admits a linear kernel.
	Our main contribution is to prove that \aPDD does not admit a polynomial kernel when parameterized by the vertex cover number plus the diversity threshold~$D$, even if the phylogenetic tree is a star.
	This implies the non-existence of a polynomial kernel for \aPDD also when parameterized by a range of structural parameters of the food web, such as its distance to cluster, treewidth, pathwidth, feedback vertex set, and others.
\end{abstract}

\section{Introduction}

Humanity is living through what biologists increasingly call the \emph{sixth
mass extinction}~\cite{barnosky2011has,cowie2022sixth}, with species
disappearing at an unprecedented pace. Beyond the tragic immediate loss of
biodiversity, the decline of entire lineages undermines ecological stability
and the evolutionary potential of life itself~\cite{cardinale2012biodiversity}.
Given that an insufficient amount of resources is given to preserving
all species, this forces conservationists to prioritize: Which
subset of species should be preserved to maximize genetic and ecological value?
A key framework for such prioritization is \emph{phylogenetic diversity}
(PD)~\cite{FAITH1992} has become a central concept in biodiversity assessment,
as it is for example the basis of the ``EDGE of existence''
programme of the Zoological Society of London~\cite{EDGE}, and has influenced
policy discussions on evolutionary distinctiveness and ecosystem
resilience~\cite{vellend2011measuring}.

Given a phylogenetic ``tree of life'' of all species, the phylogenetic
diversity of any subset of species is defined as the total length of
all edges on paths from the root of the tree to species in this subset. A
high-PD subset of species is expected to preserve a wide range of functional
and genetic traits~\cite{faith2016pd,karanth2019phylogenetic}.

Given a phylogenetic tree, a set of $k$ species that maximize PD can be selected using a
simple greedy strategy~\cite{FAITH1992,Pardi2005,steel}. This combination
of algorithmic tractability and biological relevance spurred
generalizations that attempt to model real-world ecological and temporal
constraints more faithfully~\cite{hartmann,TimePD,GNAP,pardi07}. These variants
address, among others, heterogeneous conservation costs~\cite{GNAP,pardi07},
spatially restricted habitats~\cite{bordewich2012budgeted}, or the presence of
reticulate
events~\cite{bordewichNetworks,MAPPD,van2024maximizing,van2025average,WickeFischer2018}.

Species do not exist in isolation though: Their survival often depends on prey,
hosts, or mutualistic partners. Conserving a predator without maintaining
sufficient prey populations is therefore futile~\cite{pimm1982food}. This
insight led to the introduction of the \PDDlong~(\PDD) problem~\cite{moulton},
which augments the classical model with a (directed) ecological network
encoding, for example, food-web interactions and which stipulates that a viable
selection of saved species contains at least one prey for each predator that is saved.

In a solution of \PDD that preserves only one prey per predator, the extinction
of a single species can trigger an entire extinction cascade of species that
transitively depend on its survival.\todo{JS: I question the scientific soundness of this argument. we always assume that we can preserve any species at 100\% certainty. Also in our reductions and algorithms we heavily depend on that. Thus, I would argue to remove it. NH: I dont see the problem.}
Moreover, not all prey contribute equally
to a predator's survival. Weighted food webs were introduced in
\cite{lieberman,scotti,yang2024analysis} as models that support more robust
solutions that preserve more than one prey per predator. This leads to weighted
versions of \PDD \cite{WeightedFW}, which stipulate that at least some fraction
of the prey of each preserved predator are also preserved.

\PDD is \NP-hard~\cite{faller,spillner}.  Thus, it is interesting to study \PDD
and its weighted variants through the lens of parameterized complexity theory,
which provides a framework for analyzing which aspects of a problem drive its
computational hardness~\cite{cygan,downeybook} and for developing efficient
algorithms for many problems in spite of their classical intractability.

\todosi{A comment which I would suggest to consider for the journal version, only.\\
In Page 4, the definition of "kernelization" is given formally. One or two sentences (in the introduction or another appropriate place such as related work) that give an intuitive explanation could motivate the reader as to why we are interested in a polynomial kernel.}

\paragraph*{Our contribution.}

We extend the work of~\cite{WeightedFW} by analyzing the limits of
parameterized algorithms and kernelizations for maximization of PD with food
webs. Specifically, we introduce \aPDD, which requires that at least an
$\alpha$-fraction of the prey of any preserved predator must also be preserved.
This generalizes the previously studied \fPDD and \hPDD
problems~\cite{MAPPDviability,WeightedFW} (for~$\alpha=1$ and
for~$\alpha=\nicefrac12$), as well as \PDD \cite{moulton} (for $\alpha$ approaching
zero). In \Cref{sec:k}, we prove that \fPDD is \Wh{1}-hard with respect
to both the diversity threshold~$D$ and its complement parameter~$\Dbar$, even when
the food web is directed bipartite. In our main result, in \Cref{sec:kernels},
we prove that \fPDD
admits no polynomial kernel when parameterized by the vertex cover number
plus~$D$, unless \NPcoNPpoly. This implies the non-existence of
polynomial kernels for other structural parameters, also, including treewidth,
distance to (co-)cluster, and more. In \cref{sec:reductions}, we provide
reductions from \PDD and \fPDD to \aPDD, for any fixed $\alpha \in (0,1]$,
allowing hardness results of \cref{sec:k,sec:kernels} to extend to \aPDD,
for any constant $\alpha \in (0,1]$. We complement these non-existence results
by showing that \fPDD admits a linear kernel parameterized by the distance of
the food web to a clique.

\section{Preliminaries}

\label{sec:prelims}

\subsection{Definitions}

\paragraph*{Graphs.}

We use standard graph terminology as in~\cite{diestel2025graph}. A directed
graph~$G=(V,E)$ is \emph{directed bipartite} if each vertex has only outgoing or only incoming edges.
Given a directed, acyclic graph with positive edge-weights~$\omega$, the \emph{suppression} of a degree-$2$ vertex~$v$ with incident edges~$uv$ and~$vw$ replaces $v$, $uv$, and $vw$ by an edge $uw$ of weight $\omega(uv)+\omega(vw)$.

\paragraph*{Phylogenetic trees and phylogenetic diversity.}

\looseness=-1
A directed tree~$(V,E)$ with an \emph{edge weight} function~\mbox{$\w: E \to \mathbb{N}_{>0}$}
is a \emph{(phylogenetic)~$X$-tree~$\Tree=(V,E,\w)$} if
\begin{inparaenum}[(i)]
	\item there is a unique vertex of in-degree $0$, the \emph{root}~$\rho$;
	\item each \emph{leaf} has in-degree~$1$ and out-degree~$0$;
	\item the leaves are labeled bijectively with the elements of a set~$X$ of \emph{taxa} (e.g., species); and
	\item the remaining vertices have in-degree $1$ and out-degree at least $2$.
\end{inparaenum}
A tree is a \emph{star} if every non-root vertex is a leaf.

\looseness=-1
For a phylogenetic $X$-tree $\Tree$ and a set $S \subseteq X$, the set
of edges on the paths from $\rho$ to the leaves in~$S$ is denoted as
$E_{\Tree}(S)$. The \emph{phylogenetic diversity}~$\PD(S)$ of
$S$ is defined as
\begin{equation}
  \PD(S) := \sum_{\mathclap{e \in E_{\Tree}(S)}} \w(e).
  \label{eqn:PDdef} 
\end{equation}

\looseness=-1
For any edge $e \in \Tree$, let $\off(e)$ be the set of leaves $x$ such that $e$
belongs to the path from the root of~$\Tree$ to~$x$. Then the \emph{Some-$A$-} and
\emph{All-$A$-contractions} of~$\Tree$ are $(X \setminus A)$-trees obtained
from~$\Tree$ by removing all edges $e \in \Tree$ such that~$\off(e) \cap A \ne \emptyset$ and~$\off(e) \subseteq A$, respectively,
identifying all in-degree-$0$ vertices with each other, and suppressing all degree-$2$ vertices.

\paragraph*{Food webs.}

A \emph{food web~$\Food=(X,E_\Food)$} on a set of taxa~$X$ is a
directed, acyclic graph with an optional edge weight function~$\gamma: E_\Food
\to (0,1]$.
A taxon $x$ is a \emph{prey} of a taxon $y$ and $y$ is a \emph{predator} of $x$, if~$xy \in E$.
The sets of prey and predators of~$x$ are denoted as
$\preyF{x}{\Food}$ and $\predatorsF{x}{\Food}$, respectively. We omit the
subscript if the food web is contextually clear.
Taxa without prey are \emph{sources}. The set of taxa that can reach a taxon
$x\in X$ in \Food is denoted as $X_{\le x}$, and $X_{\ge x}$ is the set of taxa that
can be reached from $x$ in~\Food.

For any~$\alpha \in (0,1]$, a set of taxa~$A\subseteq X$ is \emph{\aviable}
if~$|A\cap \prey{x}| \ge \alpha \cdot |\prey{x}|$, for all~$x\in A$. A set of
taxa~$A\subseteq X$ is \emph{\viable} if~$A\cap \prey{x}$ is non-empty, for
each non-source~$x\in A$.
If $\Food$ has edge weights~$\gamma$, a set of taxa $A$ is~\emph{$\gamma$-viable}
if~$\sum_{uv \in E_\Food:\, u \in A} \gamma(uv) \ge 1$, for each
non-source~$v\in A$.

\paragraph*{Phylogenetic diversity with dependencies.}

\looseness=-1
The main problem we study is this:

\problemdef{\aPDDlong (\aPDD)}{
  A phylogenetic~$X$-tree~$\Tree$, a food-web~$\Food$ on $X$, and integers~$k$
  and~$D$.
}{
  Is there an \aviable set~$S\subseteq X$ such that~$|S|\le k$ and
  \mbox{$\PD(S)\ge D$}?
}

\looseness=-1
In the \PDD problem, we instead ask for the existence of an \viable set. In
\sPDD and \asPDD, the input of \PDD and \aPDD restricts $\Tree$
to be a star.

\problemdef{\wPDD}{
  A phylogenetic~$X$-tree~$\Tree$, a \emph{weighted} food web~$\Food$ on $X$
    with an edge-weight function~$\gamma$, and integers~$k$ and~$D$.
}{
  Is there a~$\gamma$-viable set~$S\subseteq X$ such that $|S| \le k$ and
  \mbox{$\PD(S)\ge D$}?
}

Observe that \PDD and \aPDD are special cases of \wPDD where $\gamma (uv)$ equals
1 and $(\alpha \cdot \deg^-(v))^{-1}$, respectively, for each $uv \in E$.

\paragraph*{Parameterized complexity.}

\looseness=-1
We study \PDD and \aPDD from a parametrized perspective.
\ifConference
\else
There
exists a natural bijection between decision problems and formal languages over
some alphabet $\Sigma$.  A \emph{parameterized language} is subset $L \subseteq
\Sigma^* \times \NN$.  For a pair $(\sigma,k) \in L$, $k$ is called the
\emph{parameter} of this pair.  A parameterized language is
\emph{fixed-parameter tractable} (\FPT) if there exists an algorithm $\A$ that
decides for any input $(\sigma,k) \in \Sigma^* \times \NN$ whether $(\sigma,k)
\in L$, and does so in time $f(k) \cdot p(|\sigma|)$, where
$f(\cdot)$ is some computable function and $p(\cdot)$ is some polynomial.  A
parameterized language is \emph{slicewise polynomial} (\XP) if there exists an
algorithm $\A$ that decides for any input $(\sigma,k) \in \Sigma^* \times \NN$
whether $(\sigma,k) \in L$, and does so in time at most $|\sigma|^{p(k)}$,
where $p(\cdot)$ is some polynomial.  Every problem that is \FPT is also in \XP.
$\Wh{h}$-hard problems, for $h \ge 1$, are unlikely to be \FPT.
See~\cite{cygan,downeybook} for more detailed background on parameterized
complexity.

A \emph{compression} from a parameterized language $P \subseteq \Sigma^* \times
\NN$ to a parameterized language~$Q \subseteq \Sigma^* \times \NN$ is an
algorithm $\A$ which, given a pair $(\sigma,k) \in \Sigma^* \times \NN$,
constructs another pair~$(\tau,\ell) \in \Sigma^* \times \NN$ such that $\ell
\le k$, $|\tau| \le f(k)$, for some computable function $f(\cdot)$,
and~$(\sigma,k) \in P$ if and only if $(\tau,\ell) \in Q$.  Moreover, $\A$~must
accomplish this construction in time polynomial in $|\sigma|$.  If $P = Q$,
then we call $\A$ a \emph{kernelization}.  We say that $\A$ is a
\emph{polynomial} compression or kernel if $f(\cdot)$ is a polynomial.

\fi

The main parameters we consider are the \emph{solution size}~$k$, the
\emph{diversity threshold}~$D$, their complement parameters---the
\emph{species loss}~$\kbar := |X| - k$ and the \emph{acceptable loss of
diversity}~$\Dbar := \PD (X) - D$---, the \emph{vertex cover number}~$\vc$ of
$\Food$ (i.e., the size of a smallest subset of vertices that contains at least
one endpoint of every edge), and the \emph{distance to clique}~$d$ of $\Food$
(i.e., the minimum number of vertices to be deleted to make $\Food$ a clique).

\paragraph*{OR-cross compositions.}

Cross compositions \cite[Chapter~15]{cygan} are the standard framework for
proving that a parameterized problem $\Pi$ does not admit a polynomial kernel.
A \emph{polynomial equivalence relation} on some alphabet $\Sigma^*$ is an
equivalence relation $\sim$ such that, for polynomials $p(\cdot)$ and $q(\cdot)$,
\begin{enumerate}[(i)]
	\item~$\sim$~partitions $\Sigma^{\le n}$ into $p(n)$ equivalence classes, for
	each~$n \in \NN_{>0}$, and
	\item~It takes $q(|\sigma_1| + |\sigma_2|)$
	time to decide whether $\sigma_1 \sim \sigma_2$, for any $\sigma_1, \sigma_2 \in \Sigma^*$.
\end{enumerate}
A decision problem $\Pi_1$ \emph{(OR-)cross-composes} into a parameterized decision problem
$\Pi_2$ if
a polynomial equivalence relation~$\sim$ on~$\Sigma^*$ and an
algorithm $\A$ exist that, given $t$ equivalent (under~$\sim$) instances
$\Instance_1, \dots, \Instance_t$ of $\Pi_1$, constructs, in~$(\sum_{i=1}^t
|\sigma_i|)^{\Oh(1)}$~time, an input~$(\Instance,k)$ of~$\Pi_2$ such that
\begin{enumerate}[(i)]
	\item $k \le p(\max \{ |\sigma_1|, \ldots, |\sigma_t| \} + \log t)$, for some polynomial~$p(\cdot)$, and
	\item $(\Instance,k)$ is a \yes-instance of $\Pi_2$ if and only if $\Instance_i$ is a \yes-instance of $\Pi_1$, for some~$i \in [t]$.
\end{enumerate}

\begin{theorem}[{\cite[Theorem~15.9]{cygan}}]
  If an \NP-hard problem $\Pi_1$ cross-composes into a parameterized problem $\Pi_2$,
  then $\Pi_2$ does not admit a kernelization of polynomial size, assuming \NPcoNPpoly.
\end{theorem}

\todo{JS: In \Cref{thm:noKernel}, we use \Cref{lem:1toAlpha}\ref{itm:fewAddedVertices} as a subroutine. This however only works because it is a PPT. We have nowhere defined these. I am fine to have it as implicit right now as it is, but would definitely appreciate a definition for a journal version. NZ: Yes, I thought about this when going through the proof.  I think we may be able to handle this without introducing what a PPT is and simply pointing out in the proof that the parameter does not go up.}

\subsection{Related Work}

\ifJournal
Only a few years after proving that \MPDlong can be solved with a greedy
algorithm~\cite{FAITH1992,steel,Pardi2005},
\fi
Moulton et al.~\cite{moulton}
defined \PDD and provided efficient algorithms under various
assumptions on the input.
The conjecture that \PDD is \NP-hard~\cite{spillner} was proven by Faller et
al.~\cite{faller}, even for directed trees as food webs.  They also proved
that \sPDD is \NP-hard if the food web is a bipartite graph, and can be
solved in polynomial time if the food web is a directed tree~\cite{faller}.
\PDD is \FPT when parameterized by the solution size plus the phylogenetic tree's
height, and~\sPDD is \FPT when parameterized by the food web's treewidth~\cite{PDD}.

In \cite{WeightedFW}, \PDD was extended to require species to depend on a larger share
of their prey. For a large number of structural parameters of the food web, \fsPDD is \FPT, and \hsPDD admits \XP-algorithms~\cite{WeightedFW}.
In~\cite{MAPPDviability}, a combination of viability with phylogenetic
diversity on phylogenetic \emph{networks} has been introduced, and a complete
complexity dichotomy was presented.

\subsection{Preliminary Observations}

We end this section with two observations.

\begin{observation}
	\label{obs:fviable}
	Given a food web~\Food and a set~$S \subseteq X$, then $S$ is \fviable in
	\Food if and only if $X_{\le x} \subseteq S$, for each~$x\in S$.
\end{observation}\todo{JS: Where did the proof of this observation go? It is easily verified, but I would find at least an intuition useful.  NZ: The proof is completely obvious. NH: agreed.}

\PDD remains \NP-hard on instances where the food web is an out-tree---a spider
graph to be more precise~\cite{faller}. It has been noted that the center
vertex can be removed, leaving only direct paths~\cite{PDD}. Together with
\cref{lem:EPStoAlpha} (see \cref{sec:reductions}), this implies the following.

\begin{observation}
	\label{obs:dcw}
	\aPDD, for any $\alpha \in (0,1]$, is \NP-hard even if the directed acyclic
	cutwidth of the food web is~1.
\end{observation}

Consequently, as directed acyclic cutwidth is an upper bound on scanwidth, we
can not hope to find an efficient algorithm for \aPDD with respect to the
scanwidth of the food web~\cite{berry}.

\section{Reductions from \fsPDD and \sPDD to \asPDD}

\label{sec:reductions}

In this section, we present polynomial-time reductions from both \fsPDD and
\sPDD to \asPDD, for any fixed~$\alpha \in (0,1]$. We use these reductions in
\Cref{sec:k,sec:kernels} to extend hardness results of \fsPDD and \sPDD to
\asPDD.

\begin{lemma}
	\label{lem:1toAlpha}
  Fix any~$\alpha \in (0,1]$. Given an instance~$\Instance = (\Tree, \Food =
  (X,E_\Food), k, D)$ of \fsPDD with $|X| \ge k$, in time~$\Oh((n^2+k) / \alpha)$, we
  can compute an equivalent instance~$\Instance' = (\Tree', \Food' =
  (X',E_{\Food'})), k', D')$ of \asPDD such that
  \begin{enumerate}[(i)]
    \item $k' = k$ and $D' = D$,
	\item $X \subseteq X'$ and~$\Food'[X] = \Food$,
    \item $\PD(\{x\}) = \PDsub{\Tree'}(\{x\})$, for each~$x\in X$,
    \item $\PDsub{\Tree'}(\{x'\}) = 1$, for each~$x' \in X' \setminus X$, and
	\item Sources in~\Food are sources in~$\Food'$.
	\end{enumerate}
  Additionally, we can choose one of the following properties
  $\Food'$ satisfies:
	\begin{enumerate}[(a)]
    \item\label{itm:bipartite}If~$\Food$ is directed bipartite, then~$\Food'$
      is also directed bipartite.
    \item\label{itm:fewAddedVertices}At most~$2 \alpha^{-1} \cdot
      (k+\max_{x\in X}|\preyF{x}{\Food}|)$ new taxa are added and, if~$\Food$
      is bipartite, then~$\Food'$ is also bipartite.
	\end{enumerate}
\end{lemma}

\begin{proof}
  We define two instances~$\Instance_{\ref{itm:bipartite}} = (\Tree_a, \Food_a,
  k, D)$ and~$\Instance_{\ref{itm:fewAddedVertices}} = (\Tree_b, \Food_b, k,
  D)$ that satisfy Conditions (i)--(v) and either Condition
  (\ref{itm:bipartite}) or Condition (\ref{itm:fewAddedVertices}),
  respectively. $\Instance_{\ref{itm:fewAddedVertices}}$ is a yes-instance of
  \asPDD if and only if $\Instance$ is a yes-instance of \fsPDD.  The same is
  true for $\Instance_{\ref{itm:bipartite}}$ but only if $\Food$ is directed
  bipartite.

  To construct $\Instance_{\ref{itm:bipartite}}$ from $\Instance$, we add a set 
  $A_x := \{y_1, \ldots, y_{a_x}\}$ of
  $a_x := \lfloor (\alpha^{-1} - 1) \cdot |\preyF{x}{\Food}| \rfloor$ prey
  to each vertex $x \in X$, and we make each node $y_i \in A_x$ a child of
  $\rho$ in $\Tree$ with $\omega(\rho y_i) = 1$.

  To construct $\Instance_{\ref{itm:fewAddedVertices}}$ from $\Instance$, we
  add two sets~$A_1$ and $A_2$ of $\lfloor \max_{x\in X} |\preyF{x}{\Food}|
  \cdot \alpha^{-1} \rfloor$ taxa each, and two sets~$B_1$ and~$B_2$ of~$\lceil
  k\cdot \alpha^{-1} \rceil$ taxa each, and we add an edge $ba$ to $\Food$, for
  all $i \in \{1, 2\}$, $b \in B_i$, and $a \in A_i$. All new taxa~$y$ are
  added to $\Tree$ as children of the root $\rho$ with $\omega(\rho y) = 1$.
  Let~$X_1 \cup X_2$ be a bipartition of \Food. For each $i \in \{1,2\}$ and
  each $x \in X_i$, choose~$a_x$ arbitrary taxa of~$A_i$ as prey of~$x$.

  $\Instance_{\ref{itm:bipartite}}$ and
  $\Instance_{\ref{itm:fewAddedVertices}}$ satisfy Conditions (i)--(v),
  that $\Instance_{\ref{itm:bipartite}}$ satisfies Condition
  (\ref{itm:bipartite}), that $\Instance_{\ref{itm:fewAddedVertices}}$
  satisfies Condition (\ref{itm:fewAddedVertices}), and that both instances can
  be constructed in $\Oh((n^2+k) / \alpha)$ time.  We have to prove that
  $\Instance$ is a yes-instance of \fsPDD if and only if
  $\Instance_{\ref{itm:fewAddedVertices}}$ is a yes-instance of \asPDD, and
  that the same holds for $\Instance_{\ref{itm:bipartite}}$, provided \Food is
  directed bipartite.
	
  First, assume that $S$ is a solution to \fsPDD on~$\Instance$. Then
  $\PDsub{\Tree_a}(S) = \PDsub{\Tree_b}(S) = \PD(S) \ge D$ and $|S| \le k$.
  Thus, to show that $S$ is a solution to \asPDD on
  $\Instance_{\ref{itm:bipartite}}$
  and~$\Instance_{\ref{itm:fewAddedVertices}}$, it suffices to show that~$S$ is
  \aviable in~$\Food_{\ref{itm:bipartite}}$
  and~$\Food_{\ref{itm:fewAddedVertices}}$.  To this end, note that $|S \cap
  \preyF{x}{\Food_a}| = |S \cap \preyF{x}{\Food_b}| = |S \cap \preyF{x}{\Food}|
  = |\preyF{x}{\Food}|$, for all $x\in S$, because $S$ is \fviable in~\Food.
  However, $|\preyF{x}{\Food_a}| = |\preyF{x}{\Food_b}| = a_x +
  |\preyF{x}{\Food}| \le  |\preyF{x}{\Food}|/ \alpha$. Thus, $|S \cap
  \preyF{x}{\Food_a}| = |S \cap \preyF{x}{\Food_b}| \ge
  \alpha|\preyF{x}{\Food_a}| = \alpha|\preyF{x}{\Food_b}|$, for all $x \in S$,
  that is, $S$ is \aviable in both $\Instance_{\ref{itm:bipartite}}$
  and~$\Instance_{\ref{itm:fewAddedVertices}}$.

  Next, assume that $S$ is a solution to \asPDD on
  $\Instance_{\ref{itm:fewAddedVertices}}$. Then note that saving any vertex in
  $A_i$ requires saving at least $k$ vertices in $B_i$\todo{JS: And we may
    assume $D>k$? NH: I think we may assume $D \geq k$ (which is enough?) since
    edge weights are positive integers, so the case with $D < k$ is a trivial no
	instance.  NZ: Why is this important?  The argument does not need this.
	JS: I felt like it did.}, for
  any $i \in \{1, 2\}$.  Thus, $S \cap (A_1 \cup A_2) = \emptyset$.  Since
  $\preyF{x}{\Food_{\ref{itm:fewAddedVertices}}} \cap (B_1 \cup B_2) =
  \emptyset$, for all $x \in X$, this implies that
  $\preyF{x}{\Food_{\ref{itm:fewAddedVertices}}} \cap S \subseteq X$, for all~$x
  \in X$. In
  particular $\preyF{x}{\Food_{\ref{itm:fewAddedVertices}}} \cap S =
  \preyF{x}{\Food} \cap S$, for all $x \in X$. Thus, for all $x \in X \cap S$, we have
  $|\preyF{x}{\Food} \cap S| = |\preyF{x}{\Food_{\ref{itm:fewAddedVertices}}}
  \cap S| \ge \alpha |\preyF{x}{\Food_{\ref{itm:fewAddedVertices}}}| >
  |\preyF{x}{\Food}| - 1$, that is, $X \cap S$ is \fviable in~\Food.  Any subset $Y
  \subseteq X$ such that every edge in \Food incident to $Y$ has its tail in $Y$
  is trivially \fviable.  We call such a set a \emph{source set} in \Food. Thus,
  if we choose a set $Y \subseteq X$ such that~$S' := (X \cap S) \cup Y$ has
  size $|S'| = k$, then~$S'$, being the union of two \fviable sets, is also
  \fviable. Moreover, since every vertex $y \in S \setminus X$ satisfies
  $\omega(\rho y) = 1$, and $\omega(\rho x) \ge 1$, for all~$x \in X$, we have
  $\PD(S') \ge \PDsub{\Tree_b}(S) \ge D$.  This shows that \Instance is a
  yes-instance of~\fsPDD.

  Finally, assume that $S$ is a solution to \asPDD on
  $\Instance_{\ref{itm:bipartite}}$, and that \Food is directed bipartite.
  Then, as in the previous paragraph, $|\preyF{x}{\Food_{\ref{itm:bipartite}}}
  \cap S| > |\preyF{x}{\Food}| - 1$, for every vertex $x \in
  X \cap S$.  Thus, the set $S'' := (X \cap S) \cup \bigcup_{x \in X \cap S}
  \preyF{x}{\Food}$ is a \fviable set in $\Food$ of size $|S''| \le |S|$.  We can
  once again construct a \fviable set $S'$ of size $|S'| = k$ in $\Food$ by
  defining $S' := S'' \cup Y$, for a source set $Y$ in \Food such that $|S'' \cup
  Y| = k$.  As in the previous paragraph, this ensures that~$|S'| \ge |S|$ and,
  therefore, that $\PDsub{\Tree}(S') \ge \PDsub{\Tree_a}(S) \ge D$, as every
  vertex $y \in S \setminus X$ satisfies~$\omega(\rho y) = 1$, and $\omega(\rho
  x) \ge 1$, for all $x \in X$. This shows that \Instance is a yes-instance
  of~\fsPDD.
\end{proof}

In \cref{lem:1toAlpha}, we ensured that in the resulting instance, the solution
size~$k$ and the diversity threshold~$D$ remain unchanged, but $\kbar$ and
$\Dbar$ increased. The following reduction from \sPDD ensures that $\Dbar$
remains unchanged.

\begin{lemma}
	\label{lem:EPStoAlpha}
  Fix any~$\alpha \in (0,1)$. Given an instance~$\Instance = (\Tree, \Food =
  (X,E_\Food), k, D)$ of~\sPDD, in time~$\Oh(n^2 \cdot \alpha / (1 - \alpha))$,
  we can compute an equivalent instance~$\Instance' = (\Tree', \Food' =
  (X',E_{\Food'})), k', D')$ of \asPDD such that
  \begin{enumerate}[(i)]
	\item $\pDbar := \PDsub{\Tree'}(X') - D' = \Dbar$,
	\item $X \subseteq X'$ and~$\Food'[X] = \Food$,
	\item $\PD(\{x\}) = \PDsub{\Tree'}(\{x\})$, for each~$x\in X$,
	\item $\PDsub{\Tree'}(\{x'\}) = 2\max_{x\in X} \PD(\{x\})$, for each~$x' \in
	X' \setminus X$,
	\item Sources in~\Food are sources in~$\Food'$,
	\item Taxa~$x \in X' \setminus X$ are sources in~$\Food'$, and
	\item $\Food'$ is directed bipartite if $\Food$ is directed bipartite.
	\end{enumerate}
\end{lemma}

\begin{proof}
  To construct $\Instance'$ from $\Instance$, we add a set~$A_x := \{y_1,
  \dots, y_{a_x}\}$ of $a_x := \lceil |\preyF{x}{\Food}| \cdot \alpha / (1 -
  \alpha) \rceil - 1$ taxa as prey of~$x$, for each $x \in X$. This ensures
  that, unless $x$ is a source of $\Food$, we have $a_x <
  \alpha|\preyF{x}{\Food'}| \le a_x + 1$. We obtain~$\Tree'$ from~$\Tree$ by making
  each $y_i \in A_x$ a child of $\rho$ and defining $\omega(\rho y_i) = 2m$,
  where~$m := \max_{x\in X} \PD(\{x\})$. Let~$A := \sum_{x \in X} a_x$, $k' :=
  k + A$, and~$D' := 2D + 2m \cdot A$. This defines an instance~$\Instance' :=
  (\Tree',\Food',k',D')$ of \asPDD.

  $\Instance'$ satisfies Conditions (i)--(vii) and
  that $\Instance'$ can be constructed in $\Oh(n^2 \cdot \alpha / (1 -
  \alpha))$ time---which is $\Oh(n^2)$ for any constant~$\alpha$. It remains to
  show that $\Instance$ is a yes-instance of \sPDD if and only if $\Instance'$
  is a yes-instance of \asPDD.
	
  First, assume that $S$ is a solution to \sPDD on \Instance.  Then, we define $S'
  := S \cup \bigcup_{x \in X} A_x$.  This ensures~$\PDsub{\Tree'}(S')
  = \PD(S) + 2m A \ge D + 2mA = D'$ and $|S'| = |S| + A \le k'$.  Thus, to
  prove that $S'$ is a solution to \asPDD on $\Instance'$, we have to show that
  $S'$ is \aviable in $\Food'$.  This follows because all $y \in \bigcup_{x \in
  X} A_x$ are sources in $\Food'$ and, for all $x \in X$, we have
  \[
    |S' \cap \preyF{x}{\Food'}|
    = |S \cap \preyF{x}{\Food}| + a_x
    \ge 1 + a_x
    \ge \alpha|\preyF{x}{\Food'}|.
  \]

  Now, assume that $S$ is a solution to \asPDD on $\Instance'$. Because taxa in
  $\bigcup_{x\in X} A_x$ are sources and have maximal diversity, we may assume
  that~$\bigcup_{x\in X} A_x \subseteq S$. Now let~$S' := S \cap X$. Then
  $|S'| = |S| - A \le k$ and~\mbox{$\PDsub{\Tree}(S') = \PDsub{\Tree'}(S) - 2mA \ge
  D$}.  We need to prove that $S'$ is \viable in
  $\Food$.  So consider any non-source vertex~$x$ of~$S'$.  Since $S$ is
  \aviable in $\Food'$, we have
  \[
    |S \cap \preyF{x}{\Food'}|
    \ge \alpha|\preyF{x}{\Food'}|
    > a_x.
  \]
  Since $|\preyF{x}{\Food'} \setminus \preyF{x}{\Food}| = a_x$, this implies
  that $|S' \cap \preyF{x}{\Food}| = |S \cap \preyF{x}{\Food}| > 0$, that is,
  $S'$ is \viable in $\Food$.
\end{proof}

\section{Parameterization by Solution Size and Diversity}

\label{sec:k}

In this section, we show that, when parameterized by $k$, $D$, $\kbar$ or
$\Dbar$, \wPDD is \XP and \asPDD is \Wh{1}-hard, for every~$\alpha \in (0,1]$.
The hardness results even hold for directed bipartite food webs.

Recall that if there is a solution for an instance of~\wPDD, then one of size
\emph{exactly}~$k$ exists~\cite{WeightedFW}. Thus, it suffices to search for a
solution to \wPDD on a given instance \Instance among the subsets of $X$ of
size $k$.  Since there are $\binom{|X|}{k}$ such subsets, and it takes $\Oh(n +
m)$ time to test whether any such subset $S \subseteq X$ is $\gamma$-viable and
whether it has diversity at least $D$, we have:

\begin{observation}
	\label{cor:k-XP}
	\wPDD is \XP when parameterized by $k$ or $\kbar$.
\end{observation}

This result transfers to~\PDD and to~\aPDD for each~$\alpha \in (0,1]$.
Because~$k\le D$ and~$\kbar \le \Dbar$, \cref{cor:k-XP} also holds for the
bigger parameters~$D$ and~$\Dbar$. In the remainder of this section, we prove
the \Wh{1}-hardness of \asPDD:

\begin{theorem}
	\label{thm:D-Wh}
  \asPDD, for any~$\alpha \in (0,1]$, is \Wh{1}-hard when parameterized by~$D$,
  even if the food web is directed bipartite and
  the maximum edge weight in the phylogenetic tree is $2$.
\end{theorem}

This result is somewhat surprising, because \sPDD is \FPT when parameterized
by~$D$~\cite{PDD}. Recall that in \PDD, one needs to save at least one prey to
save a non-source species. This is equivalent to choosing~$\alpha$ as the
inverse of the maximum in-degree of the food web.  Thus, fixing $\alpha$ to be
a constant, as done in~\asPDD, impacts the complexity of the problem.

\begin{proof}
  Given a pair $(G,k)$, where $G = (V,E)$ is an undirected graph and $k \in
  \NN$, the \CP problem asks whether $G$ has a clique of size $k$. Given such a
  an instance~$(G,k)$ of \CP, we construct from it an instance $\Instance' = (\Tree', \Food',
  k', D')$ of \fsPDD such that $D' = k^2$ and~$(G,k)$ is a yes-instance of
  clique if and only if~$\Instance'$ is a yes-instance of \fsPDD. Since \CP is
  \Wh{1}-hard with respect to~$k$~\cite{cygan}, this shows that \fsPDD is
  \Wh{1}-hard with respect to~$D$. The \Wh{1}-hardness of \asPDD with respect
  to~$D$, for any $\alpha \in (0,1]$, then follows from
  \cref{lem:1toAlpha}\ref{itm:bipartite}.

  Given an instance $\Instance = (G = (V,E), k)$ of \CP, we define the
  instance~$\Instance'$ of \fsPDD as follows: We divide every edge $e$ of $G$
  into two edges. Let $v_e$ be the introduced vertex, let~$V_E$ be the set
  of all these vertices introduced by subdividing all edges of $G$, and let
  $G'$ be the resulting graph.  We define a food web $\Food'$ by directing all
  edges of $G'$ towards the vertices in $V_E$. We define a star tree~$\Tree'$
  with root~$\rho$, $\w(\rho v) := 1$, for all~$v\in V$, and~$\w(\rho v_e) :=
  2$, for all~$v_e\in V_E$.
  Finally, let $\Instance' := (\Tree', \Food', k', D')$ be an instance of \fsPDD,
  where~$k' := \binom{k}{2} + k$ and~\mbox{$D' := 2\binom{k}{2} + k = k^2$}.
  Clearly, $\Instance'$ can be constructed from $\Instance$ in polynomial time.
  We need to prove that $\Instance$ is a yes-instance if and only if
  $\Instance'$ is.
	
  First, assume that $S$ is a clique of size~$k$ in~$G$. Then~$G[S]$
  contains~$\binom{k}{2}$ edges. Consequently, in $\Instance'$, the set~$S'$
  containing the vertices in~$S$ and the vertices associated with the edges
  in~$G[S]$ has size~$k'$.  The diversity of this set is~$D'$. The set~$S'$ is
  \fviable in~$\Food'$ because both endpoints of each edge in $G[S]$ are
  in~$S$.

  Now, assume that $S$ is a solution to \fsPDD on $\Instance'$.  If $S$
  contains $\ell$ vertices from~$V$, then it contains at most~$\binom{\ell}{2}$
  vertices from~$V_E$. Consequently, since~$\PD(S) \ge D'$, $S$ contains~$k$
  vertices in~$V$ and~$\binom{k}{2}$ vertices from~$V_E$. This implies that $S
  \cap V$ is a clique of size~$k$ from~$G$.\todo{NZ: This paragraph is very dense
  	now. JS: What do you mean by dense?}
\end{proof}

Next, we show that \asPDD, for any~$\alpha \in (0,1]$, is also \Wh{1}-hard when
parameterized by the complement parameter~$\Dbar$.

\begin{theorem}
	\label{thm:Dbar-Wh}
  \asPDD, for any~$\alpha \in (0,1]$, is \Wh{1}-hard when parameterized
  by~$\Dbar$, even if the food
  web is directed bipartite and the maximum edge weight in the phylogenetic tree is $2$.
\end{theorem}

\begin{proof}
  As \sPDD is \Wh{1}-hard when parameterized by $\Dbar$~\cite{PDD},
  \cref{lem:EPStoAlpha} proves the result for all $\alpha \in (0,1)$. The
  result for $\alpha = 1$, that is, for \fsPDD, remains to show.  Similarly, as in
  the proof of \cref{thm:D-Wh}, we reduce from \CP.  Here, the
  instance~$\Instance' = (\Tree',\Food',k',D')$ of \fsPDD that we construct from an
  instance~$\Instance = (G = (V,E),k)$ of \CP satisfies $\Dbar = O(k^2)$.
  Thus, since \CP is \Wh{1}-hard with respect to $k$, \fsPDD is \Wh{1}-hard
  with respect to $\Dbar$.

  To obtain $\Instance'$ from $\Instance$, we define a graph $G'$, vertices
  $v_e$, for all edges $e \in E$, and the set~$V_E$ as in the proof of
  \cref{thm:D-Wh}. Unlike in the proof of \cref{thm:D-Wh}, we construct
  $\Food'$ from~$G'$ by directing all edges of $G'$ towards $V$. We define a
  star tree~$\Tree'$ with root~$\rho$, and we set $\w(\rho v) = 2$, for
  all~$v\in V$, and $\w(\rho v_e) = 1$, for all $v_e\in V_E$. Then $\Instance'
  := (\Tree', G', k', D')$, where~$k' = |V| + |V_E| - \binom{k}{2} - k$ and~$D'
  = 2|V| + |V_E| - \binom{k}{2} - 2k$. In particular, $\pkbar = |V(G')| - k' =
  \binom{k}{2} + k$ and~$\pDbar = \PDsub{\Tree'}(V(G')) - D' = \binom{k}{2} +
  2k$.

  Instance~$\Instance'$ can be constructed in polynomial
  time.  It remains to prove that~$\Instance$ is a yes-instance of \CP if and only
  if~$\Instance'$ is a yes-instance of~\fsPDD.

  First, assume that $S$ is a clique of size~$k$ in~$G$. Then~$G[S]$
  contains~$\binom{k}{2}$ edges. Let the set~$S'$
  contain all vertices \emph{except} the vertices in $S$ and the vertices
  associated with the edges in $G[S]$. Then $S'$ has size $k'$ and a diversity of $D'$.
  All vertices in $V_E$ are sources and every vertex $v \in V$
  with $\preyF{v}{\Food'} \not\subseteq S'$ is in~$S$. Thus, $S'$ is
  \fviable.

  Now, assume that $S$ is a solution to \fsPDD on $\Instance'$. For each
  vertex~$v_e \in V_E$ that is not in~$S$, neither endpoint of $e$ is in $S$.
  Thus, if there are $\binom{\ell}{2}$ vertices in~$V_E$ that are not in~$S$,
  then there are at least~$\ell$ vertices in $V$ that are not in~$S$.
  Consequently, since~$|S| \le |V| + |V_E| - \binom{k}{2} - k$ but
  $\PDsub{\Tree'}(S) \ge 2|V| + |V_E| - \binom{k}{2} - 2k$, $S$ misses exactly~$k$
  vertices from $V$ and~$\binom{k}{2}$ vertices from $V_E$, representing all
  edges in $G$ between the vertices in $V \setminus S$. Thus, $V \setminus
  S$ is a clique of size~$k$ in~$G$.\todo{NZ: Again, very dense now. JS: Great! Or not?}
\end{proof}

\section{Kernels}

\label{sec:kernels}

\definecolor{r}{rgb}{1.0, 0.75, 0.75}  
\definecolor{a}{rgb}{1.0, 0.85, 0.55}  
\definecolor{b}{rgb}{0.70, 0.85, 1.0}  
\definecolor{g}{rgb}{0.75, 0.95, 0.75} 
\definecolor{w}{rgb}{1, 1, 1}
\definecolor{grey}{rgb}{0.9453, 0.9453, 0.9453}
\definecolor{grey2}{rgb}{0.85, 0.85, 0.85}

\begin{figure}[t]
	\tikzstyle{para}=[rectangle,draw=black,minimum height=.8cm,fill=gray!10,rounded corners=1mm, on grid]
	
	\newcommand{\tworows}[2]{\begin{tabular}{c}{#1}\\{#2}\end{tabular}}
	\newcommand{\threerows}[3]{\begin{tabular}{c}{#1}\\{#2}\\{#3}\end{tabular}}
	\newcommand{\distto}[1]{\tworows{Distance to}{#1}}
	\newcommand{\disttoc}[2]{\threerows{Distance to}{#1}{#2}}
	
	\DeclareRobustCommand{\tikzdot}[1]{\tikz[baseline=-0.6ex]{\node[draw,fill=#1,inner sep=2pt,circle] at (0,0) {};}}
	\DeclareRobustCommand{\tikzdottc}[2]{\tikz[baseline=-0.6ex]{\node[draw,diagonal fill={#1}{#2},inner sep=2pt,circle] at (0,0) {};}}
	
	\tikzset{
		diagonal fill/.style 2 args={fill=#2, path picture={
				\fill[#1, sharp corners] (path picture bounding box.south west) -|
				(path picture bounding box.north east) -- cycle;}},
		reversed diagonal fill/.style 2 args={fill=#2, path picture={
				\fill[#1, sharp corners] (path picture bounding box.north west) |- 
				(path picture bounding box.south east) -- cycle;}}
	}
	\centering
	\begin{tikzpicture}[node distance=2*0.45cm and 3.7*0.38cm, every node/.style={scale=0.57}]
		\linespread{1}
		\node[para,fill=a] (vc) {Minimum Vertex Cover};
		\node[para, diagonal fill=br, xshift=38mm] (ml) [right=of vc] {Max Leaf \#};
		\node[para, xshift=-40mm,fill=g] (dc) [left=of vc] {\distto{Clique}};
		
		\node[para, diagonal fill=ar, xshift=11mm] (dcc) [below= of dc] {\distto{Cluster}}
		edge[-{Stealth}] (dc)
		edge[-{Stealth}, bend left=10] (vc);
		\node[para, fill=a,xshift=41mm] (dcl) [below= of dc] {\distto{Co-Cluster}}
		edge[-{Stealth}] (dc)
		edge[-{Stealth}] (vc);
		\node[para, diagonal fill=ar, xshift=8mm] (ddp) [below=of vc] {\distto{Disjoint Paths}}
		edge[-{Stealth}] (vc)
		edge[-{Stealth}] (ml);
		\node[para,diagonal fill=br] (fes) [below =of ml] {\tworows{Feedback}{Edge Set}}
		edge[-{Stealth}] (ml);
		\node[para, xshift=2mm, diagonal fill=br] (bw) [below right=of ml] {Bandwidth}
		edge[-{Stealth}] (ml);
		\node[para, xshift=5mm, yshift=0mm,diagonal fill=ar] (td) [right=of ddp] {Treedepth}
		edge[-{Stealth}, bend right=28] (vc);
		
		\node[para, diagonal fill=ar] (fvs) [below= of ddp] {\tworows{Feedback}{Vertex Set}}
		edge[-{Stealth}] (ddp)
		edge[-{Stealth}, bend right=5] (fes);
		\node[para, diagonal fill=br] (dcw) [right= of bw] {\threerows{Directed}{Acyclic}{Cutwidth}};
		\node[para, diagonal fill=br] (cw) [below= of bw] {Cutwidth}
		edge[-{Stealth}] (bw)
		edge[-{Stealth}] (dcw);
		\node[para, xshift=-30mm, diagonal fill=ar] (pw) [below= of cw] {Pathwidth}
		edge[-{Stealth}] (ddp)
		edge[-{Stealth}] (td)
		edge[-{Stealth}] (cw);
		\node[para, diagonal fill=br] (sw) [below right= of cw] {Scanwidth}
		edge[-{Stealth}] (dcw);
		
		\node[para, xshift=5mm, fill=r] (dbp) [below left= of fvs] {\distto{Bipartite}}
		edge[-{Stealth}] (fvs);
		\node[para, yshift=0mm, diagonal fill=ar] (tw) [below=of pw] {Treewidth}
		edge[-{Stealth}] (fvs)
		edge[-{Stealth}] (pw)
		edge[-{Stealth}, bend right=8] (sw);

		\node[para, yshift=-20mm, diagonal fill={yellow!50}{grey2}] [below= of dc] {\tworows{{\fPDD}\;\;\;\;\;\;\;\;\;\;\;}{\;\;\;\;\;\;\;\;\;\;\fsPDD}};
	\end{tikzpicture}
  \caption{The relationship between structural parameters
    of the food web and the complexity of solving \fPDD (top left) and \fsPDD
    (bottom right). A parameter~$p$ is marked red (\tikzdot{r}) if the problem
    is \NP-hard for constant values of~$p$. If the problem is \FPT with respect
    to~$p$, then $p$ is marked in orange (\tikzdot{a}), green (\tikzdot{g}) or
    blue (\tikzdot{b}).  Green means that the problem has a polynomial kernel,
    orange means that it does not, and blue means that kernelization remains
    open.
    An edge~$p_1 p_2$ between parameters~$p_1$ and~$p_2$ indicates that, in every graph,
    the parameter~$p_1$ can be bounded by a function of~$p_2$. A
    more in-depth look into the hierarchy of graph parameters can be found
    in~\cite{graphparameters}.}
	\label{fig:fPDD-results}
\end{figure}
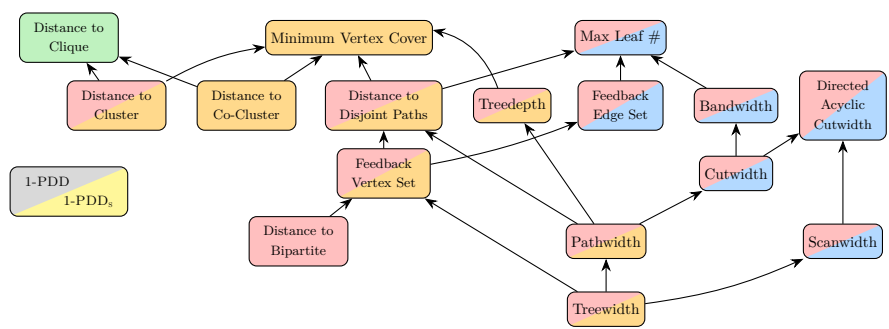

While~\rwPDD, a generalization of~\PDD and~\aPDD, for any~$\alpha \in (0,1]$,
is \FPT with respect to the vertex cover number~$\vc$~\cite{WeightedFW}, we
prove in this section that~\asPDD, for any~$\alpha \in (0,1]$, does not admit a
polynomial kernel when parameterized by~$\vc + D$, even on bipartite food webs,
assuming~\NPcoNPpoly.
This result is already known for~\sPDD~\cite{PDD}.\footnote{In~\cite{PDD}, it
is only stated that~\sPDD does not have a polynomial kernel for $D$, but a
closer look reveals that the larger parameter~$\vc + D$ is also excluded.}
Further, we show that \fPDD admits a kernel with respect to the food web's
distance to clique. Whether this result can be generalized to \aPDD remains
open. In fact, it is unknown even whether \hPDD can be solved in polynomial
time on cliques~\cite{WeightedFW}.
These results and their implications are visualized in \Cref{fig:fPDD-results}.

\subsection{No Polynomial Kernel With Respect to Vertex Cover Number}

The following is the main result of this paper.

\begin{theorem}
	\label{thm:noKernel}
  \asPDD, for any~$\alpha \in (0,1]$, does not admit a poly\-no\-mi\-al-size
  kernel when parameterized by $\vc + D$, the food web's vertex cover number
  plus the diversity threshold, even if the food web is bipartite,
  unless $\textup{NP} \subseteq \textup{coNP} / \textup{poly}$.
\end{theorem}

\begin{figure}[t]
	\centering
	\begin{tikzpicture}[
		x=2.7mm,y=2.7mm,
		vertex/.style={fill,circle,inner sep=0pt,minimum size=2.5pt},
		dot/.style={fill,circle,inner sep=0pt,minimum size=1pt},
		edge/.style={draw},
		component/.style={fill=black!20,rounded corners=3pt,inner sep=6pt},
		vc/.style={draw,densely dashed,rounded corners=3pt,inner sep=3pt}
		]
		\path       node [vertex,label] (v0) {}
		++(0:1)   node [vertex]       (v1) {}
		++(0:1)   node [vertex]       (v2) {}
		++(0:1)   node [dot] {}
		++(0:0.5) node [dot] {}
		++(0:0.5) node [dot] {}
		++(0:1)   node [vertex]       (vn) {};
		\foreach \i in {0,...,3} {
			\foreach \k in {0,1} {
				\path (0:10)
				++(0:8*\i+4*\k) node [vertex,label=right:$\k_\i$] (b\i\k) {}
				++(-1.25,2)     node [vertex] (b\i\k-0) {}
				++(0:0.5)       node [vertex] (b\i\k-1) {}
				++(0:0.5)       node [vertex] (b\i\k-2) {}
				++(0:0.5)       node [dot] {}
				++(0:0.25)      node [dot] {}
				++(0:0.25)      node [dot] {}
				++(0:0.5)       node [vertex] (b\i\k-n) {};
				\path [edge] (b\i\k)
				foreach \j in {0,1,2,n} {
					edge (b\i\k-\j)
				};
				\path
				(b\i\k-0.west) +(0,3pt) coordinate (l)
				(b\i\k-n.east) +(0,3pt) coordinate (r);
				\draw [decorate,decoration=brace] (l) to node [above=3pt,inner sep=0pt] (k2) {$k^2 - 1$} (r);
			}
		}
		\foreach \i/\d in {0/0,1/10,2/20,t/37} {
			\path 
			(\d-1.375,-10) node [vertex] (e\i-0) {}
			++(0:1)   node [vertex] (e\i-1) {}
			++(0:1)   node [vertex] (e\i-2) {}
			++(0:1)   node [dot] {}
			++(0:0.5) node [dot] {}
			++(0:0.5) node [dot] {}
			++(0:1)   node [vertex] (e\i-n) {};
		}
		\path
		($(e2-n)!0.5!(et-0)$) node [dot] {}
		+(180:1)              node [dot] {}
		+(0:1)                node [dot] {};
		\draw [densely dotted] foreach \i in {0,2,t} {
			foreach \j in {0,1,2,n} {
				(e\i-\j) edge +(90:2) edge +(85:2) edge +(95:2)
			}
		}
		foreach \v in {v0,v2,b00,b11,b21,b31} {
			(\v) edge +(270:2) edge +(265:2) edge +(275:2)
		}
		foreach \v in {e1-0,e1-1,e1-n} {
			(\v) edge +(90:2) edge +(85:2) edge +(95:2)
		}
		(v1)  edge +(270:2) edge +(265:2)
		(vn)  edge +(270:2) edge +(265:2)
		(b01) edge +(275:2) edge +(270:2)
		(b10) edge +(275:2) edge +(270:2)
		(b20) edge +(275:2) edge +(270:2)
		(b30) edge +(275:2) edge +(270:2);
		\begin{scope}[decoration={
				markings,
				mark=at position .5 with {\arrowreversed{Stealth}}
			}]
			\path [edge,postaction={decorate}] (e1-2) .. controls +(75:7) and +(265:7) .. (b30);
			\path [edge,postaction={decorate}] (e1-2) .. controls +(95:7) and +(275:7) .. (v1);
			\path [edge,postaction={decorate}] (e1-2) .. controls +(90:7) and +(275:7) .. (vn);
			\path [edge,postaction={decorate}] (e1-2) .. controls +(85:7) and +(265:7) .. (b01);
			\path [edge,postaction={decorate}] (e1-2) .. controls +(80:7) and +(265:7) .. (b10);
			\path [edge,postaction={decorate}] (e1-2) .. controls +(75:7) and +(265:7) .. (b20);
			\path [edge,postaction={decorate}] (e1-2) .. controls +(75:7) and +(265:7) .. (b30);
		\end{scope}
		\path
		(e1-2.south) node [below=2pt,inner sep=0pt] (e) {$\{u,v\}$}
		(v1.north)   node [above=2pt,inner sep=0pt] (u) {$u$}
		(vn.north)   node [above=2pt,inner sep=0pt] (v) {$v$};
		\begin{scope}[on background layer]
			\node [component,fit={(b00) (b00-0) (b31-n) (k2.north)}] (B) {};
			\node [component,fit={(v0) (vn) (v0 |- k2.north)}] (V) {};
			\node [component,fit={(e0-0) (e0-n) (e0-0 |- e.south)}] (E0) {};
			\node [component,fit={(e1-0) (e1-n) (e1-0 |- e.south)}] (E1) {};
			\node [component,fit={(e2-0) (e2-n) (e2-0 |- e.south)}] (E2) {};
			\node [component,fit={(et-0) (et-n) (et-0 |- e.south)}] (Et) {};
		\end{scope}
		\node [vc,fit={(V.west |- v0) (B.east |- v0) (B.south) (u) (v)}] (vc) {};
		\node [above right] at (vc.south east) {$C$};
		\node [above right] at (E0.south east) {$E_0$};
		\node [above right] at (E1.south east) {$E_1$};
		\node [above right] at (E2.south east) {$E_2$};
		\node [above left] at (Et.south west) {$E_{t-1}$};
		\node [below left] at (V.north west)  {$V$};
		\node [below right] at (B.north east)  {$B'$};
	\end{tikzpicture}
	\caption{The cross-composition from $t \in \{9,\ldots,16\}$ clique instances
		to 1-PDD.  As illustrated for one vertex $\{u, v\} \in E_1$, every vertex
		$\{x, y\} \in E_i$, for $i \in [t-1]_0$, has in-edges from the endpoints $x, y
		\in V$ of the corresponding edge, and from those vertices among $0_0, 1_0,
		\ldots, 0_{\log t}, 1_{\log t}$ that encode the index $i$.}
	\label{fig:kernel}
\end{figure}
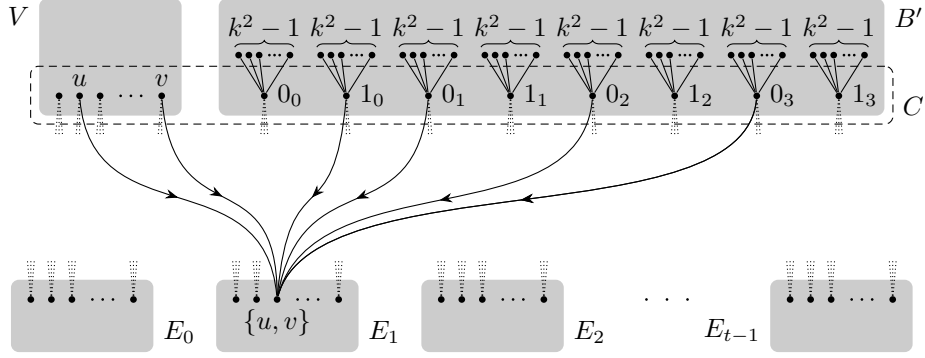

\begin{proof}
  We use an OR-cross composition from \CP to \fsPDD to prove the result for
  \fsPDD.  The result for all $\alpha \in (0,1]$ then follows by applying
  \cref{lem:1toAlpha}\ref{itm:fewAddedVertices}.

  As polynomial equivalence relation, we consider two \CP instances $(G_1 =
  (V_1,E_1), k_1)$ and $(G_2 = (V_2,E_2), k_2)$ to be equivalent if $|V_1| = |V_2|$
  and $k_1 = k_2$.  By renaming vertices, we can
  assume that $V_1 = V_2$ for equivalent instances.

  Given a set $\mathfrak{I} := \{(G_i = (V,E_i), k) \mid i\in [t-1]_0\}$ of $t$
  equivalent instances of \CP, let $n := \max_{i \in [t-1]_0} |G_i|$ and $N :=
  \sum_{i \in [t - 1]_0} |G_i|$. We construct the following instance
  $\Instance' = (\Tree', \Food', k', D')$ of \fsPDD from the set $\mathfrak{I}$ of instances:
  \todos{Why is the one instance with mathfrak and the other with mathcal. Can we make taht uniform? I would suggest mathcal.}
  
  The vertex set of $\Food'$ contains the vertices in~$V$, one vertex $e$ for
  every edge $e \in E := E_0 \uplus \cdots \uplus E_{t-1}$, and a set~$B$
  of~$2\ell$ special vertices $0_0, 1_0, \ldots, 0_{\ell-1}, 1_{\ell-1}$, where
  $\ell = \lceil \log t \rceil$.  Each such vertex~$b_i \in B$ has a set~$B_{b_i}$
  of~$k^2 - 1$ prey not connected to any other vertex.  Let
  $B' := B \cup \bigcup_{v \in B} B_v$.  For each edge $e = \{u,v\} \in E_b$,
  let~$b_{\ell-1}b_{\ell-2}\cdots b_0$ be the binary representation of~$b$.
  Then, let vertex~$e$ have $u, v \in V$ and $b_0, \ldots, b_{\ell-1} \in B$
  as prey in $\Food'$.
  (See \cref{fig:kernel} for an illustration.)
  
  Let $\Tree'$ be a star with root $\rho$, we set $\w(\rho e) = 2$, for each $e \in
  E$, and $\w(\rho v)$, for each $v \in V \cup B'$, and let $k' := k^2 \cdot
  (\ell + 1) - \binom{k}{2}$ and $D' := k^2 \cdot (\ell + 1)$.
  This completes the definition of~$\Instance'$.

  This construction of $\Instance'$ from $\mathfrak{I}$ can be carried out
  in time polynomial in $N$. $\Food'$~is bipartite, with bipartitions
  $V \cup B$ and $E \cup \bigcup_{v \in B} B_v$. Since
  $D' \in O(k^2 \log t) \subseteq O(n^2 \log t)$ and $C := V \cup B$
  is a vertex cover of $\Food'$ of size $O(n + \log
  t)$, to show that it constitutes an OR-cross composition from \CP to \fsPDD
  parameterized by $\vc + D'$, it remains to show that there exists a graph
  $G_b$, $b \in [t-1]_0$, which has a clique of size $k$ if and only if
  $\Instance'$ is a yes-instance of \fsPDD.

  First, assume that $S$ is a clique of size $k$ in $G_b$.  Then let $S'
  \subseteq V(\Food')$ be the set containing all vertices in $S$, all vertices
  in $E_b$ representing edges between vertices in $S$, and all vertices in
  $\bigcup_{i=0}^{\ell - 1} \{b_i\} \cup B_{b_i}$, where
  $b_{\ell-1}b_{\ell-2}\cdots b_0$ is the binary representation of $b$.  This
  set $S'$ is \fviable because all vertices in $V \cup \bigcup_{i=0}^{\ell - 1}
  B_{b_i} $ are sources, for each $b_i \in S'$, all vertices in~$B_{b_i}$ are
  in $S'$, and for each $\{u,v\} \in E_b \cap S'$, we have $u, v, b_0, \ldots,
  b_{\ell - 1} \in S'$.  The size of~$S'$ is~$\binom{k}{2} + k + k^2\ell =
  k^2(\ell + 1) - \binom{k}{2} = k'$, and its diversity is $2\binom{k}{2} + k +
  k^2\ell = k^2(\ell + 1) = D'$.  Thus, $S'$ is a solution to~\fsPDD on $\Instance'$.

  Next, assume that $S$ is a solution to \fsPDD on $\Instance'$.  Since $D' > k'$,
  $S$~and~$E$ intersect.  Thus, $|S \cap B| \ge \ell$, as every
  vertex in $E$ has $\ell$ prey in $B$. On the other hand, $S$ also
  contains all vertices in $B_v$, for every vertex $v \in S \cap B$.  Thus, $S$~contains
  at least $k^2 \cdot|S \cap B|$ taxa.  Since $|S| \le k' < k^2(\ell + 1)$, this shows that~$|S
  \cap B| = \ell$.  Since $S$ contains vertices in~$E$, there exists an index~$b
  \in [t - 1]_0$ such that $S$ and $E_b$ intersect. Thus, the $\ell$
  vertices in~$S \cap B$ must be the binary encoding
  $b_{\ell-1}b_{\ell-2}\cdots b_0$ of $b$, and we have $S \cap E = S \cap E_b$.
  Now, let $S'$ be the intersection of~$S$ and~$V \cup E_b$.  Then, the size of~$S'$ is $k^2\ell \le
  \binom{k}{2} + k$ and its diversity $\PDsub{\Tree'}(S) - k^2\ell \ge
  2\binom{k}{2} + k$.  Thus, $|S' \cap E_b| \ge \binom{k}{2}$ and $|S' \cap V|
  \le k$.  However, $S'$ contains the
  two vertices $u$ and~$v$, for every edge $\{u,v\} \in S' \cap E_b$.
  Thus, the constraints that $|S'| \le \binom{k}{2} +
  k$, $|S' \cap E_b| \ge \binom{k}{2}$, and $|S' \cap V| \le k$ imply that~$|S'
  \cap E_b| = \binom{k}{2}$ and $|S' \cap V| = k$. Therefore, $G_b[S' \cap V]$ is a
  clique of size~$k$ because~$S' \cap E_b$ contains $\binom{k}{2}$ edges
  between the $k$ vertices in~$S' \cap V$. Hence, $(G_b,k) \in \mathfrak{I}$
  is a yes-instance of \CP.
\end{proof}

\subsection{Kernel for Distance to Clique}

In this section, we show the existence of a kernel for \fPDD when parameterized
by the distance to clique~$d$ of the food web.

\begin{theorem}
	\label{pps:kernel-clique-dist}
  Given an instance~$\Instance = (\Tree,\Food,k,D)$ of~\fPDD and a
  set~$Z\subseteq X$, such that~$\Food-Z$ is a clique, in~$\Oh(n+m)$~time, we
  can compute an equivalent instance of~\fPDD with at most~$2 \cdot |Z|$ taxa.
\end{theorem}

\begin{proof}
  We use two reduction rules, each of which we apply once, to construct the
  reduced instance $\Instance'$ from $\Instance$.  These rules are illustrated
  in \cref{fig:reduction-rules}.

  Since $\Food[X \setminus Z]$ is a clique and acyclic, it has a unique
  topological ordering~$\tau$.  For any~$x \in X \setminus Z$, we have $\{y \in
  X \setminus Z \mid \tau(y) \le \tau(x)\} \subseteq X_{\le x}$.  By
  \cref{obs:fviable}, any solution~$S$ to \fPDD on $\Instance$ that contains
  $x$ must also include $X_{\le x}$.  Thus, no vertex $x \in X \setminus Z$
  with $\tau(x) \ge k + 1$ can be part of such a solution, nor can any vertices
  in $X_{\ge x}$.  Moreover, if~$z$ is the vertex in $X \setminus Z$
  with~$\tau(z) = k + 1$, then $x \in X_{\ge z}$ and, therefore, $X_{\ge x}
  \subseteq X_{\ge z}$, for any vertex $x \in X \setminus Z$ with $\tau(x) \ge
  k + 1$.  This motivates the first reduction rule:

  \begin{rrule}
    \label{rr:d-1}
    If $z \in X \setminus Z$ satisfies $\tau(z) = k+1$, then remove all
    vertices in $X_{\ge z}$ from $\Food$ and replace $\Tree$ with the
    all-$X_{\ge z}$-contraction of~\Tree.
  \end{rrule}

  If $\Instance' = (\Tree', \Food', k, D)$ is the instance obtained by applying
  this rule to $\Instance$, then we just observed that none of the vertices
  removed from $\Food$ by this rule can be part of a \fviable set~$S$ in~$\Food$.
  Thus, $S$ is also \fviable in $\Food'$. Conversely, for every
  vertex $x$ of $\Food'$, all vertices that can reach $x$ in $\Food$ are in
  $\Food'$.  Thus, any \fviable set in $\Food'$ is also \fviable in $\Food$.
  This shows that $\Food$ and $\Food'$ have the same \fviable sets.  Since
  $\Tree'$ is the all-$X_{\ge z}$-contraction of~\Tree, any subset $S$
  of the vertices in $\Food'$ satisfies $\PDsub{\Tree}(S) = \PDsub{\Tree'}(S)$.
  Thus, $\Instance$ is a yes-instance if and only if $\Instance'$ is.  This
  proves that \cref{rr:d-1} is safe.

  Now consider any solution $S$ to \fPDD on $\Instance$.  We can assume that
  $|S| = k$~\cite{WeightedFW}. Since~$|S \cap Z| \le |Z|$, we have $|S \cap (X
  \setminus Z)| \ge k - |Z|$.  Thus, the first $k - |Z|$ vertices in $\tau$
  must be in $S$.  This motivates the second reduction rule:

  \begin{rrule}
    \label{rr:d-2}
    If $x\in X \setminus Z$ satisfies $\tau(x) = k - |Z|$, then remove all
    vertices in $X_{\le x}$ from $\Food$, replace $\Tree$ with the some-$X_{\le
    x}$-contraction of~\Tree, reduce~$k$ by~$|X_{\le x}|$, and
    reduce $D$ by~$\PD(X_{\le x})$.
  \end{rrule}

\begin{figure}[t]
	\centering
	\begin{tikzpicture}[
		x=6.5mm,y=6.5mm,
		vertex/.style={fill,circle,inner sep=0pt,minimum size=2.5pt},
		dot/.style={fill,circle,inner sep=0pt,minimum size=1pt},
		edge/.style={draw},
		component/.style={fill=black!20,rounded corners=3pt,inner xsep=9pt,inner ysep=6pt},
		part/.style={draw,densely dashed,rounded corners=3pt,inner xsep=6pt,inner ysep=9pt},
		vc/.style={draw,densely dashed,rounded corners=3pt,inner sep=3pt},
		every label/.append style={inner sep=0pt,label distance=3pt}
		]
		\path node [vertex] (a) {}
		+(270:3) node [vertex] (za) {}
		++(0:1) node [vertex] (b) {}
		++(0:3) node [vertex] (x) {}
		++(0:1.5) node [vertex] (x') {}
		+(270:3) node [vertex] (zx') {}
		++(0:1) node [vertex] (x'') {}
		++(0:3) node [vertex] (z') {}
		+(270:3) node [vertex] (zz') {}
		++(0:1.5) node [vertex] (z) {}
		++(0:1) node [vertex] (z'') {}
		++(0:3) node [vertex] (c) {}
		++(0:1) node [vertex] (d) {}
		+(270:3) node [vertex] (zd) {}
		(z'') +(1.5,-3) node [vertex] (zc) {}
		(b) +(1.5,-3) node [vertex] (zb) {}
		(x'') +(1.5,-3) node [vertex] (zx'') {}
		;
		\foreach \u/\v in {a/b,x/x',x'/x'',z'/z,z/z'',c/d} {
			\path [edge,-Stealth] (\u) -- (\v);
		}
		\draw [densely dotted] foreach \v in {a,b,x,x',x'',z',z,z''} {
			(\v) edge +(20:0.5) edge +(30:0.5) edge +(40:0.5)
		};
		\draw [densely dotted] foreach \v in {x,x',x'',z',z,z'',c,d} {
			(\v) edge +(160:0.5) edge +(150:0.5) edge +(140:0.5)
		};
		\draw
		(b) -- +(0:1)
		(x'') -- +(0:1)
		(z'') -- +(0:1);
		\draw [Stealth-] (z') -- +(180:1);
		\draw [Stealth-] (c)  -- +(180:1);
		\draw [Stealth-] (x)  -- +(180:1);
		\draw [-Stealth] (a) -- (zb);
		\draw [-Stealth] (zb) -- (x');
		\draw [-Stealth] (zb) -- (z);
		\draw [-Stealth] (zb) -- (x);
		\draw [-Stealth] (za) -- (x);
		\draw [-Stealth] (za) -- (b);
		\draw [-Stealth] (b) -- (zx');
		\draw [-Stealth] (x') -- (zx'');
		\draw [-Stealth] (zx') -- (z'');
		\draw [-Stealth] (x'') -- (zz');
		\draw [-Stealth] (zc) -- (c);
		\draw [-Stealth] (z'') -- (zc);
		\draw [-Stealth] (c) -- (zd);
		\path foreach \v in {b,x'',z''} {
			(\v)
			+(0:1.25) node [dot] {}
			+(0:1.5) node [dot] {}
			+(0:1.75) node [dot] {}
		};
		\node [inner sep=0pt,above left,yshift=12pt] (lx) at (x) {$\tau(x) = k - |Z|$};
		\node [inner sep=0pt,above right,yshift=12pt] (lz) at (z) {$\tau(z) = k + 1$};
		\draw [ultra thin] (x) -- ([xshift=-6pt]lx.south east) (z) -- ([xshift=6pt]lz.south west);
		\begin{scope}[on background layer]
			\node [component,fit={(a) (d) (lx) (lz)},label=left:$X \setminus Z$] {};
			\node [component,fit={(za) (zd)},label=left:$Z$] {};
		\end{scope}
		\node [part,fit={(a) (za) (x |- lx.north)},label=above:{Removed by \cref{rr:d-1}}] {};
		\node [part,fit={(x') (zx') (z' |- lx.north)},label=above:$\Food'$] {};
		\node [part,fit={(d) (zd) (z |- lx.north)},label=above:{Removed by \cref{rr:d-2}}] {};
	\end{tikzpicture}
	\caption{Illustration of the reduction rules used in the proof of
		\cref{pps:kernel-clique-dist}.}
	\label{fig:reduction-rules}
\end{figure}
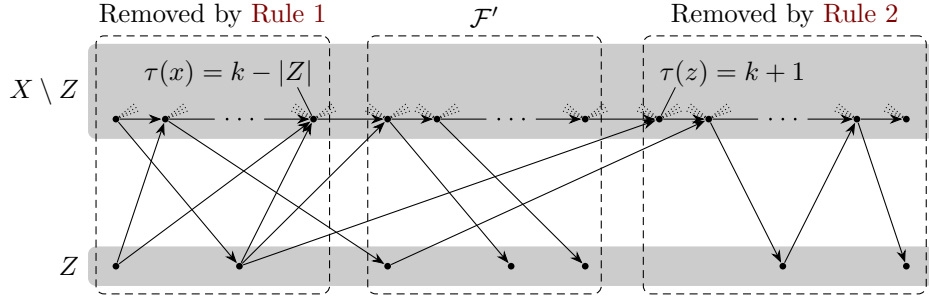

  Let $\Instance' = (\Tree', \Food', k', D')$ be the instance obtained by
  applying~\Cref{rr:d-2} to~$\Instance$.  We observed that any solution $S$ to
  \fPDD on $\Instance$ must contain $x$ and, therefore, all vertices in $X_{\le
  x}$.  Thus, if $S' := S \setminus X_{\le x}$, then $|S'| \le k - |X_{\le x}|
  = k'$, and $S'$ is~\fviable in $\Food$.  Since $\Tree'$ is the
  some-$X_{\le x}$-contraction of~\Tree, we have $\PDsub{\Tree'}(S') =
  \PD(S) - \PD(X_{\le x}) \ge D'$.  Thus, $S'$ is a solution to \fPDD on~$\Instance'$.

  Conversely, if $S'$ is a solution to \fPDD on $\Instance'$, then $S := S'
  \cup X_{\le x}$ satisfies~$|S| \le |S'| + |X_{\le x}| \le k$ and $\PD(S) =
  \PDsub{\Tree'}(S') + \PD(X_{\le x}) \ge D$. For every vertex $z \in S$, every
  vertex in $X_{\le z}$ is either in $S'$ (if it is in $\Food'$) or in $X_{\le
  x}$ (if it is not in $\Food'$). Thus, $S$ is a solution to \fPDD on
  $\Instance$.  This proves that $\Instance$ is a yes-instance if and only if
  $\Instance'$ is a yes-instance, that is, \cref{rr:d-2} is safe.

  We obtain the desired instance $\Instance'$ from $\Instance$ by
  a single application of \Cref{rr:d-1,rr:d-2}, each.
  These reductions can be implemented in $O(n + m)$ time
  because all that is required is to compute a topological ordering~$\tau$
  of~$\Food[X \setminus Z]$, to identify the vertices $z$ and $x$ with~$\tau(z) =
  k + 1$ and~$\tau(x) = k - |Z|$, and then removing the vertices in $X_{\ge z}$
  and $X_{\le x}$ from \Food, along with counting these vertices and adjusting
  $k$ and $D$ accordingly.  Since we have shown that \cref{rr:d-1,rr:d-2} are
  safe, $\Instance$ is a yes-instance if and only if $\Instance'$ is a
  yes-instance.  Thus, to finish the proof of the theorem, we need to prove
  that $\Food'$ has at most $2 \cdot |Z|$ vertices.  However, $\Food'$ contains only
  those vertices $y$ in $X \setminus Z$ that satisfy $k - |Z| < \tau(y) \le k$.
  There are exactly $|Z|$ such vertices.  In addition, $\Food'$~contains at
  most all vertics in $Z$.  This is an additional $|Z|$ vertices, so $\Food'$
  has at most $2 \cdot |Z|$ vertices.
\end{proof}

\section{Discussion}

\label{sec:discussion}

In this paper, we considered the problem of saving $k$ species of maximum
phylogenetic diversity subject to the constraint that at least an
$\alpha$-fraction of the prey of every saved species is also saved. This
problem had been studied before for~$\alpha \in \{1/2,1\}$, while we allow any
value~$\alpha \in (0,1]$ in this paper.

Our main contributions are \Wh{1}-hardness results for \asPDD
parameterized by the solution size $k$, the diversity threshold $D$, and their
complements, as well as a proof that~\asPDD does not admit
a polynomial kernel with respect to $D$ plus the vertex cover number of the
food web, which implies that \asPDD does not admit a polynomial kernel with
respect to a range of structural parameters of the food web visualized in
\cref{fig:fPDD-results}.  Together, these results paint a fairly comprehensive
picture of the parameterized complexity landscape of~\asPDD.

The most interesting open question is whether \asPDD admits a polynomial kernel
with respect to the food web's scanwidth, or at least its directed acyclic
cutwidth.  These questions are marked open in \cref{fig:fPDD-results}.

Our hardness results for \asPDD were obtained via reductions from~\fsPDD and~\sPDD.
A natural question is for which structural parameters of the food web
\asPDD is harder to solve than \fsPDD or \sPDD.
For example, \hsPDD is \Wh{1}-hard with respect to the treewidth of the food
web~\cite{twvssw}, while for the same parameter, \sPDD~\cite{PDD}
and~\fsPDD~\cite{WeightedFW} are \FPT.

\bibliography{ref}

@inproceedings{PDD,
	title = {{Maximizing Phylogenetic Diversity under Ecological Constraints: A Parameterized Complexity Study}},
	author = {Komusiewicz, Christian and Schestag, Jannik},
	booktitle = {Proceedings of the 44th IARCS Annual Conference on Foundations of Software Technology and Theoretical Computer Science (FSTTCS~2024)},
	year = {2024},
	pages = {28:1--28:18},
	organization = {Schloss-Dagstuhl-Leibniz Zentrum f{\"u}r Informatik},
	doi = {10.4230/LIPIcs.FSTTCS.2024.28}
}

@inproceedings{WeightedFW,
	title = {{Weighted Food Webs Make Computing Phylogenetic Diversity So Much Harder}}, 
	author = {Jannik Schestag},
	booktitle = {Proceedings of the 51st International Conference on Current Trends in Theory and Practice of Computer Science (SOFSEM 2026)},
	TDODOpages = {},
	year = {2026},
	organization = {Springer},
	eprint = {2510.05911},
	archivePrefix = {arXiv},
	TODOdoi = {}
}

@unpublished{twvssw,
	title = {{The First Known Problem That Is FPT with Respect to Node Scanwidth but Not Treewidth}},
	author = {Jannik Schestag and Norbert Zeh},
	year = {2026},
	archivePrefix={arXiv},
	jorunal={arXiv preprint},
	eprint = {2602.06903},
	TODOdoi = {}
}

@article{faller,
	title = {{Optimizing Phylogenetic Diversity with Ecological Constraints}},
	author = {Faller, Be{\'a}ta and Semple, Charles and Welsh, Dominic},
	journal = {Annals of Combinatorics},
	volume = {15},
	number = {2},
	pages = {255--266},
	year = {2011},
	publisher = {Springer},
	doi = {10.1007/s00026-011-0093-6}
}

@article{moulton,
	title = {{Optimizing phylogenetic diversity under constraints}},
	author = {Moulton, Vincent and Semple, Charles and Steel, Mike},
	journal = {Journal of Theoretical Biology},
	volume = {246},
	number = {1},
	pages = {186--194},
	year = {2007},
	publisher = {Elsevier},
	doi = {10.1016/j.jtbi.2006.12.021}
}

@article{spillner,
	title = {{Computing Phylogenetic Diversity for Split Systems}},
	author = {Spillner, Andreas and Nguyen, Binh T. and Moulton, Vincent},
	journal = {IEEE/ACM Transactions on Computational Biology and Bioinformatics},
	volume = {5},
	number = {2},
	pages = {235--244},
	year = {2008},
	publisher = {IEEE},
	doi = {10.1109/tcbb.2007.70260}
}

@inproceedings{MAPPDviability,
	title = {{Parameterized Algorithms for Diversity of Networks with Ecological Dependencies}},
	author = {Jones, Mark and Schestag, Jannik},
	booktitle = {Proceedings of the 20th International Symposium on Parameterized and Exact Computation (IPEC 2025)},
	pages = {11:1--11:21},
	year = {2025},
	organization = {Schloss-Dagstuhl-Leibniz Zentrum f{\"u}r Informatik},
	doi = {10.4230/LIPIcs.IPEC.2025.11}
}

@article{GNAP,
	title = {{A Multivariate Complexity Analysis of the Generalized Noah's Ark Problem}},
	author = {Komusiewicz, Christian and Schestag, Jannik},
	journal = {Discrete Applied Mathematics},
	volume = {382},
	pages = {137--154},
	year = {2025},
	publisher = {Elsevier},
	doi = {10.1016/j.dam.2025.11.037}
}

@article{TimePD,
	title = {{Maximizing Phylogenetic Diversity under Time Pressure: Planning with Extinctions Ahead}},
	author = {Jones, Mark and Schestag, Jannik},
	journal = {arXiv preprint arXiv:2403.14217},
	year = {2024},
	eprint = {2403.14217},
	archivePrefix = {arXiv},
	TODOdoi = {}
}

@article{WickeFischer2018,
	title = {Phylogenetic diversity and biodiversity indices on phylogenetic networks},
	author = {Wicke, Kristina and Fischer, Mareike},
	journal = {Mathematical Biosciences},
	volume = {298},
	pages = {80-90},
	year = {2018},
	doi = {10.1016/j.mbs.2018.02.005}
}

@inproceedings{van2024maximizing,
	title = {{Phylogenetic Network Diversity Parameterized by Reticulation Number and Beyond}},
	author = {van Iersel, Leo and Jones, Mark and Schestag, Jannik and Scornavacca, Celine and Weller, Mathias},
	booktitle = {Proceedings of the 22nd RECOMB International Workshop on Comparative Genomics (RECOMB-CG 2025)},
	pages = {107--130},
	year = {2025},
	organization = {Springer},
	doi = {10.1007/978-3-031-94928-9\_7}
}

@inproceedings{van2025average,
	title = {{Average-Tree Phylogenetic Diversity of Networks}},
	author = {van Iersel, Leo and Schestag, Jannik and Jones, Mark and Scornavacca, Celine and Weller, Mathias},
	booktitle = {Proceedings of the 25th International Workshop on Algorithms in Bioinformatics (WABI 2025)},
	pages = {14:1--14:21},
	year = {2025},
	organization = {Schloss Dagstuhl--Leibniz-Zentrum f{\"u}r Informatik},
	doi = {10.4230/LIPIcs.WABI.2025.15}
}

@article{bordewichNetworks,
	title = {On the complexity of optimising variants of phylogenetic diversity on phylogenetic networks},
	author = {Bordewich, Magnus and Semple, Charles and Wicke, Kristina},
	journal = {Theoretical Computer Science},
	volume = {917},
	pages = {66--80},
	year = {2022},
	publisher = {Elsevier},
	doi = {10.1016/j.tcs.2022.03.012}
}

@inproceedings{MAPPD,
	title = {{How Can We Maximize Phylogenetic Diversity? Parameterized Approaches for Networks}},
	author = {Jones, Mark and Schestag, Jannik},
	booktitle = {Proceedings of the 18th International Symposium on Parameterized and Exact Computation (IPEC 2023)},
	pages = {30:1--30:12},
	year = {2023},
	organization = {Schloss-Dagstuhl-Leibniz Zentrum f{\"u}r Informatik},
	doi = {10.4230/LIPIcs.IPEC.2023.30}
}

@article{steel,
	title = {{Phylogenetic Diversity and the Greedy Algorithm}},
	author = {Steel, Mike},
	journal = {Systematic Biology},
	volume = {54},
	number = {4},
	pages = {527--529},
	year = {2005},
	publisher = {Society of Systematic Zoology},
	doi = {10.1080/10635150590947023}
}

@article{Pardi2005,
	title = {{Species Choice for Comparative Genomics: Being Greedy Works}},
	author = {Pardi, Fabio and Goldman, Nick},
	journal = {PLoS Genetics},
	year = {2005},
	volume = {1},
	doi = {10.1371/journal.pgen.0010071}
}

@article{pardi07,
	title = {{Resource-Aware Taxon Selection for Maximizing Phylogenetic Diversity}},
	author = {Pardi, Fabio and Goldman, Nick},
	journal = {Systematic Biology},
	volume = {56},
	number = {3},
	pages = {431--444},
	year = {2007},
	publisher = {Society of Systematic Zoology},
	doi = {10.1080/10635150701411279}
}

@article{hartmann,
	title={{Maximizing Phylogenetic Diversity in Biodiversity Conservation: Greedy Solutions to the Noah's Ark Problem}},
	author={Hartmann, Klaas and Steel, Mike},
	journal={Systematic Biology},
	volume={55},
	number={4},
	pages={644--651},
	year={2006},
	publisher={Society of Systematic Zoology},
	doi = {10.1080/10635150600873876}
}

@article{bordewich2012budgeted,
	title={{Budgeted Nature Reserve Selection with diversity feature loss and arbitrary split systems}},
	author={Bordewich, Magnus and Semple, Charles},
	journal={Journal of Mathematical Biology},
	volume={64},
	number={1},
	pages={69--85},
	year={2012},
	publisher={Springer},
	doi = {10.1007/s00285-011-0405-9}
}

@inproceedings{berry,
	title={{Scanning Phylogenetic Networks Is NP-hard}},
	author={Berry, Vincent and Scornavacca, Celine and Weller, Mathias},
	booktitle={Proceedings of the 46th International Conference on Current Trends in Theory and Practice of Computer Science (SOFSEM 2020)},
	pages={519--530},
	year={2020},
	organization={Springer},
	doi = {10.1007/978-3-030-38919-2\_42}
}

@article{FAITH1992,
	title = {{Conservation evaluation and phylogenetic diversity}},
	author = {Faith, Daniel P.},
	journal = {Biological Conservation},
	volume = {61},
	number = {1},
	pages = {1-10},
	year = {1992},
	issn = {0006-3207},
	doi = {10.1016/0006-3207(92)91201-3},
}

@article{karanth2019phylogenetic,
	title={Phylogenetic diversity as a measure of biodiversity: pros and cons},
	author={Karanth, K. Praveen and Gautam, Srishti and Arekar, Kunal and Divya, B.},
	journal={Journal of the Bombay Natural History Society},
	volume={116},
	pages={53--61},
	year={2019},
	doi = {10.17087/jbnhs/2019/v116/120848}
}

@article{vellend2011measuring,
	title={Measuring phylogenetic biodiversity},
	author={Vellend, Mark and Cornwell, William K. and Magnuson-Ford, Karen and Mooers, Arne {\O}.},
	journal={Biological diversity: frontiers in measurement and assessment},
	volume={14},
	pages={194--207},
	year={2011},
	publisher={Oxford University Press Oxford}
}

@article{faith2016pd,
	title={{The PD Phylogenetic Diversity Framework: Linking Evolutionary History to Feature Diversity for Biodiversity Conservation}},
	author={Faith, Daniel P.},
	journal={Biodiversity Conservation and Phylogenetic Systematics: Preserving our evolutionary heritage in an extinction crisis},
	pages={39--56},
	year={2016},
	publisher={Springer International Publishing},
	doi = {10.1007/978-3-319-22461-9\_3}
}

@article{EDGE,
	title={{EDGE} of existence and phylogenetic diversity},
	author={Faith, Daniel P.},
	journal={Animal Conservation},
	volume={22},
	pages={537--538},
	year={2019},
	doi = {10.1111/acv.12552}
}

@article{scotti,
	title={{Weighting, scale dependence and indirect effects in ecological networks: A comparative study}},
	author={Scotti, Marco and Podani, J{\'a}nos and Jord{\'a}n, Ferenc},
	journal={Ecological Complexity},
	volume={4},
	number={3},
	pages={148--159},
	year={2007},
	publisher={Elsevier},
	doi = {10.1016/j.ecocom.2007.05.002}
}

@article{yang2024analysis,
	title={Analysis of keystone species in a quantitative network perspective based on stable isotopes},
	author={Yang, Ruijing and Feng, Minquan and Liu, Zimeng and Wang, Xuyan and Qu, Zili},
	journal={Ecological Complexity},
	volume={59},
	pages={101092},
	year={2024},
	publisher={Elsevier},
	doi = {10.1016/j.ecocom.2024.101092}
}

@article{lieberman,
	title={Evolutionary dynamics on graphs},
	author={Lieberman, Erez and Hauert, Christoph and Nowak, Martin A.},
	journal={Nature},
	volume={433},
	number={7023},
	pages={312--316},
	year={2005},
	publisher={Nature Publishing Group UK London},
	doi = {10.1038/nature03204}
}

@book{pimm1982food,
	title={Food webs},
	author={Pimm, Stuart L.},
	year={1982},
	publisher={Springer},
	doi = {10.1007/978-94-009-5925-5\_1}
}

@article{cardinale2012biodiversity,
	title={Biodiversity loss and its impact on humanity},
	author={Cardinale, Bradley J. and Duffy, J. Emmett and Gonzalez, Andrew and others},
	journal={Nature},
	volume={486},
	number={7401},
	pages={59--67},
	year={2012},
	publisher={Nature Publishing Group},
	doi = {10.1038/nature11148}
}

@article{barnosky2011has,
	title={{Has the Earth’s sixth mass extinction already arrived?}},
	author={Barnosky, Anthony D. and Matzke, Nicholas and Tomiya, Susumu and others},
	journal={Nature},
	volume={471},
	number={7336},
	pages={51--57},
	year={2011},
	publisher={Nature Publishing Group},
	doi = {10.1038/nature09678}
}

@article{cowie2022sixth,
	title={{The Sixth Mass Extinction: fact, fiction or speculation?}},
	author={Cowie, Robert H. and Bouchet, Philippe and Fontaine, Beno{\^\i}t},
	journal={Biological Reviews},
	volume={97},
	number={2},
	pages={640--663},
	year={2022},
	publisher={Wiley Online Library},
	doi = {10.1111/brv.12816}
}

@book{cygan,
	title = {Parameterized {A}lgorithms},
	author = {Cygan,Marek and Fomin, Fedor V. and Kowalik, Lukasz and Lokshtanov, Daniel and Marx, D{\'{a}}niel and Pilipczuk, Marcin and Pilipczuk, Michal and Saurabh, Saket},
	publisher = {Springer},
	year = {2015},
	doi = {10.1007/978-3-319-21275-3}
}

@book{downeybook,
	title = {Fundamentals of {P}arameterized {C}omplexity},
	author = {Downey, Rodney G. and Fellows, Michael R.},
	series = {Texts in Computer Science},
	publisher = {Springer},
	year = {2013},
	doi = {10.1007/978-1-4471-5559-1}
}

@unpublished{graphparameters,
	title = {{The Graph Parameter Hierarchy}},
	author = {Sorge, Manuel and Weller, Mathias and Foucaud, Florent and Such{\`y}, Ond{\v{r}}ej and Ochem, Pascal and Vatshelle, Martin and Woeginger, Gerhard J.},
	note = {{URL}: \url{https://manyu.pro/assets/parameter-hierarchy.pdf}},
	year = {2020}
}

@book{diestel2025graph,
	title={{Graph Theory}},
	author={Diestel, Reinhard},
	volume={173},
	year={2025},
	publisher={Springer Nature},
	doi = {10.1007/978-3-662-70107-2}
}
\bibliographystyle{plainurl}

\end{document}